\documentclass[reqno,11pt]{amsart}
\usepackage{amsmath, latexsym, amsfonts, amssymb}
\usepackage{mathrsfs,enumerate}
\usepackage{graphics,epsfig,color}

\numberwithin{equation}{section}

\newtheorem{theorem}{Theorem}[section]
\newtheorem{lemma}[theorem]{Lemma}

\newcommand{\cC}{{\ensuremath{\mathcal C}} }

\newcommand{\cR}{{\ensuremath{\mathcal R}} }

\newcommand{\bbE}{{\ensuremath{\mathbb E}} }

\newcommand{\bbN}{{\ensuremath{\mathbb N}} }

\newcommand{\bbP}{{\ensuremath{\mathbb P}} }

\newcommand{\bbR}{{\ensuremath{\mathbb R}} }

\newcommand{\bbZ}{{\ensuremath{\mathbb Z}} }

\newfont{\indic}{bbmss12}

\title[Macroscopic diffusion from a Hamilton-like dynamics]
      {Macroscopic diffusion from a Hamilton-like dynamics}

\author[R.\ Lefevere]{Rapha\"el Lefevere}
\address{Laboratoire de Probabilit\'es
 et Mod\`eles Al\'eatoires (CNRS UMR 7599), Universit\'e Paris 7
 -- Denis Diderot, UFR Math\'ematiques, Case 7012, B\^atiment
 Chevaleret, 75205 Paris Cedex 13, France}
\email{lefevere\@@math.univ-paris-diderot.fr}

\begin{document}

\begin{abstract}
We introduce and analyze a model for the transport of particles or energy in extended lattice systems.  The dynamics of the model acts on a discrete phase space at discrete times but has nonetheless some of the characteristic properties of Hamiltonian dynamics in a confined phase space : it is  deterministic, periodic, reversible and conservative. Randomness enters the model as a way to model ignorance about initial conditions and interactions between the components of the system.  The orbits of the particles are non-intersecting random loops.  We prove, by a weak law of large number, the validity of a diffusion equation for the macroscopic observables of interest for times that are arbitrary large, but small compared to the minimal recurrence time of the dynamics.
\end{abstract}

\maketitle

Fick's law of diffusion or Fourier's law of heat conduction describe phenomena which are part of everyday life : think of the diffusion of sugar in a cup of coffee or the exponential decay in time of the temperature of the same hot coffee. The laws of microscopic physics possess features that makes them look contradictory at first sight with the phenomenological laws of macroscopic physics.  In particular, microscopic laws are reversible and when the dynamics is confined in phase space, the Poincar\'e recurrence theorem ensures quasi-periodicity of the orbits. This a well-known problem which, in various guises,  has generated a large number of debates among physicists and mathematicians alike, see for instance \cite{Bricmont,Lebowitz} for excellent reviews.  Boltzmann made a decisive contribution to the issue by insisting on the fact that in large systems, the usual laws of macroscopic physics correspond only to a {\it typical} (with respect to the initial conditions) behaviour and not a uniform one.  One should note that in general the argument can not rest on typicality of the initial conditions alone. There exist quite a few systems for which macroscopic laws of normal diffusion are not obeyed : Fourier's law is not observed in non-interacting gases, lattices of harmonic oscillators or the Toda lattice.
Still, given its pervasiveness, one can not expect ``normal" macroscopic behaviour to depend on delicate microscopic dynamical properties.  On the contrary, in the absence of any information on microscopic initial conditions or interactions reflected in a large-scale structure, Fick's law or Fourier's law seems to be the typical behaviour in Nature.
Therefore, Boltzmann's solution to the problem should  rather refer  to typical initial conditions {\it and} typical dynamics.

We propose  an instructive toy model for Fick's law or Fourier's law observed in systems which have a lattice spatial structure at the microscopic level .  This model is deterministic, reversible, periodic, conservative and amenable to a full rigorous treatment.   Randomness enters the model as a way to model ignorance about microscopic initial conditions and microscopic interactions between the components of the system. With large probability with respect to this randomness, as the number of components increase, one can show that a Boltzmann equation holds and is equivalent to a macroscopic diffusion.  The key point is that the minimal recurrence time of the dynamics increases as the number of components does.  This prevents any ``recollision" occurring for times smaller than the minimal recurrence time.  The model is reminiscent of the 
Kac ring model \cite{Kac} but possesses a rich orbit structure that makes it more similar to Hamiltonian dynamics.  The orbits of the particles are non-intersecting random loops.
\section{The dynamics and its first properties}
Le us consider a finite interval in $\bbZ$ :
$$
\Lambda_N=\{i\in\bbZ:|i|\leq N\}.
$$
To each site of $i\in\Lambda_N$, we attach a ring $\cR_i$ carrying $R$ sites $k\in\{1,\ldots,R\}$.
The model consists of particles moving on 
$$
\cC_N=\prod_{i\in\Lambda_N}\cR_i=\{(k,i):k\in\{1,\ldots,R\},\;i\in\Lambda_N\}.
$$

The second ingredient of the model is the presence of ``scatterers" that are located in between pairs of sites $(k,i)$ and $(k,i+1)$. We define variables $\xi(k,i)$ taking values in $\{0,1\}$ and $\xi(k,i)=1$ if and only if there is a scatterer between sites $(k,i)$ and $(k,i+1)$.

We put particles on the sites of $\cC_N$ and a site carries at most one particle. At fixed discrete times, all particles move forward on the rings. A particle located at site  $(k,i)$, namely at site $k$ on the ring $\cR_i$  will jump to site $k+1$ on ring $\cR_{i+1}$ (resp. $\cR_{i-1}$), if and only if the following conditions are simultaneously satisfied.
\begin{enumerate}
\item There is  a scatterer between $(k,i)$ and $(k,i+ 1)$ (resp. $(k,i- 1)$), namely $\xi(k,i)=1$ (resp. $\xi(k,i-1)=1$).
\item There are no other scatterers around that pair. 
\end{enumerate}
\begin{figure}[thb]
\includegraphics[width = .89\textwidth]{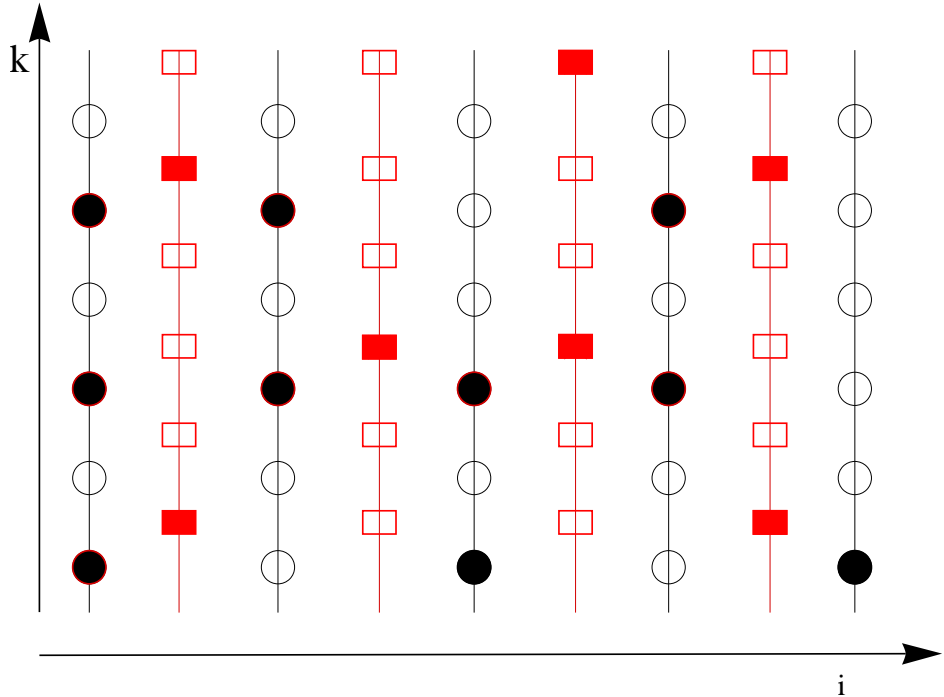} 
\caption{A configuration of particles (black disks) and scatterers (rectangles) on five rings. Periodic boundary conditions are imposed on the vertical direction.}
\end{figure}

In every other case, the particle located at site $(k,i)$ simply moves forward on its own ring to $(k+1,i)$.
Formally, we define the map $\tau:\cC_N\to\cC_N$ giving the one-step evolution of a particle located at $(k,i)\in\cC_N$ :
\begin{eqnarray}
\tau(k,i)&=&J(k,i)(k+1,i+1)+J(k,i-1)(k+1,i-1)\nonumber\\
&+&(1-J(k,i))(1-J(k,i-1))(k+1,i)
\label{taudef}
\end{eqnarray}
where 
\begin{equation}
J(k,i)=\xi(k,i)(1-\xi(k,i-1))(1-\xi(k,i+1))
\label{defJ}
\end{equation}
and we set 
$$
\xi(k,-N-1)=\xi(k,N)=0 .
$$
whenever they appear in the above expressions.
Here and in the following addition and substraction on the first  component of a point in $\cC_N$ are understood to be modulo $R$.   The first factor in (\ref{defJ}) accounts for the presence or absence of a scatterer, the two others are there to guarantee that $\tau$ is an injective map and has thus orbits similar to Hamiltonian dynamics.

To each $(k,i)\in\cC_N$, one can attach an occupation variable $\sigma(k,i)\in\{0,1\}$.
The map $\tau:\cC_N\to\cC_N$ allows to define the evolution of those variables by the relation :
\begin{equation}
\sigma(k,i;t)=\sigma(\tau^{-t}(k,i);0),\; t\in\bbN^*.
\label{ocup}
\end{equation}
This is equivalent to the recursion relation :
\begin{eqnarray}
\sigma(k,i;t)&=&(1-J(k-1,i))(1-J(k-1,i-1))\sigma(k-1,i;t-1)\nonumber\\
&+&J(k-1,i-1)\sigma(k-1,i-1;t-1)+J(k-1,i)\sigma(k-1,i+1;t-1).\nonumber\\
\label{recurevol}
\end{eqnarray}
We note that at any time $t\in\bbN^*$, the configuration $\sigma(\cdot,t)$ is obtained as a permutation of the initial occupation variables $\sigma(\cdot;0)$. Thus the dynamics is {\it conservative}.

\begin{lemma}
$\tau$ is a well-defined bijective map of $\cC_N$ into $\cC_N$.  $\tau$ is therefore invertible.
\label{taumap}
\end{lemma}
\begin{proof}
We first note that for any configuration $\xi$ of scatterers $J(k,i)\in\{0,1\}$. Next only one of the three terms in (\ref{taudef}) is non-zero.  Indeed, assume that $J(k,i)=1$ then the third term is obviously zero but so is also the second one because one must have $\xi(k,i)=1$ and thus $J(k,i-1) =\xi(k-1,i)(1-\xi(k,i-2))(1-\xi(k,i))=0$. Other cases are treated similarly.  Thus, for any $\xi$, $\tau(k,i)\in\{(k+1,i+1),(k+1,i-1),(k+1,i)\}\subset\cC_N$ and $\tau$ is well-defined. We prove now that $\tau$ is injective. Let us assume that we have $x$ and $x'$ such that $\tau(x)=\tau(x')=(k,i)$, then both $x$ and $x'$ must belong to $\{(k-1,i),(k-1,i-1),(k-1,i+1)\}$.  Suppose that $x=(k-1,i)$, then because of the definition (\ref{taudef}), we must have $J(k-1,i)=J(k-1,i-1)=0$. This implies that $\tau(k-1,i-1)\neq (k,i)$ and $\tau(k-1,i+1)\neq (k,i)$ and thus that $x=x'=(k-1,i)$.  Now, if $x=(k-1,i-1)$ and $\tau(x)=(k,i)$, then $J(k-1,i-1)=1$ and $J(k-1,i)=0$.  Thus $\xi(k-1,i-1)=1$ and $\xi(k-1,i)=0$ and therefore $\tau(k-1,i)=(k,i-1)$ and $\tau(k-1,i+1)\neq(k,i)$. Thus, in that case also, if $x'$ is such that $\tau(x')=\tau(x)$ then $x'=x$. The last possible case (when $x=(k-1,i+1)$) is treated in a similar way.  We have thus proven that $\tau:\cC_N\to\cC_N$ is an injective map. It is obvious that $\tau$ is onto, we have thus shown that $\tau$ is a bijection.
\end{proof}
One could add a degree of freedom (``velocity") to the model by deciding that all particles move either in the positive (as above) or negative direction on the rings. We would then have two maps $\tau_+$ and $\tau_-$ with $\tau_+=\tau$ and obviously $\tau_-=\tau^{-1}$. The dynamics is thus {\it reversible} in a analogous sense to Hamilton dynamics.

To each $x\in \cC_N$, we associate an orbit 
 $$
 O(x)=\{z\in\cC_N: \exists n\in\bbN\,\;\tau^n(x)=z\}
 $$
 and a period
 $$
 T(x)=\inf\{n>0:\tau^n(x)=x\}.
 $$
Since the map $\tau$ is injective, the orbits of two different sites are either identical or do not intersect.
Moreover, any orbit is self-avoiding.
\begin{lemma}Every point of $\cC_N$ is periodic (i.e. has a finite period)  and we have,
$$
R\leq T(x)\leq R(2N+1),\;\;\forall x\in\cC_N.
$$

\end{lemma}
\begin{proof}
 Since the cardinality of $\cC_N$ is finite, the fact that $\tau$ is injective implies that every orbit is periodic.  This is actually a special case of the Poincar\'e recurrence theorem and we do not repeat the proof.  An upper bound on the period is given by the number of sites contained in $\cC_N$.  A lower bound is obtained by observing that at each step one moves one step ahead on the rings.  Thus, before coming back to the same site one must have performed at least one full rotation and this gives the lower bound on the period.
 \end{proof}

\noindent We sum up some obvious but crucial facts in the following lemma, for further reference.
\begin{lemma}
\begin{enumerate}
\item $\tau^{-t}(k,i)$ depends on the configuration of scatterers only through the set of variables $\{\xi(k,i):\xi(k-n,i),\; n=1,\ldots, t,\;i=-N,\ldots,N\}$. 
\item $\sigma(k,i;t)$ is a function of $\{\xi(k,i):\xi(k-n,i),\; n=1,\ldots, t,\;i=-N,\ldots,N\}$ and $\sigma(\cdot,0)$.
\end{enumerate}
\label{dependance}
\end{lemma}
To summarize : for {\it every} configuration of scatterers, the dynamics has all the characteristic of  a Hamiltonian dynamics in a confined domain : it is deterministic,  reversible, (quasi-) periodic and conservative.  
\section{Randomness and diffusion equation}
We let randomness enter the dynamics by taking a random configuration of scatterers and random initial occupation variables.  Then $\tau:\cC_N\to\cC_N$ becomes a random map and (\ref{ocup}) implies that the configuration $\sigma(\cdot,t)$ at any time is simply a random permutation of a collection of random variables.  Likewise, orbits become non-intersecting random loops on $\cC_N$.
Scatterers stand for interactions between individual components, namely the rings.  In the absence of any macroscopic information on those interactions except its average ``density", the most natural choice is to take the scatterers $\xi(k,i)$ to be a set of independent Bernoulli variables of parameter $0<\mu<1$.  
Similar considerations lead to take the initial configuration of occupation variables as a set of independent Bernoulli random variables but with a ring-dependent parameter to account for an initial inhomogeneous distribution of density of particles.
We define the empirical density of particles on the rings $\cR_i$ which are our main ``macroscopic" quantities :
\begin{equation}
\rho^R(i,t):=\frac 1 R\sum_{k=1}^{R}\sigma(k,i;t).
\end{equation}
Our goal is to show that for times $t$ smaller than the size $R$ of the rings, this empirical density follows the solution of a diffusion equation, with large probability as $R$ grows to infinity.
As we indicate in the last section, it is easy to exhibit configurations of scatterers that would lead to an anomalous behaviour.

Let $\hat \rho(i,t)$ be a solution of the system of discrete-time evolution equations :
\begin{equation}
\left\{\begin{array}{lll}
&\hat\rho(i,t)-\hat\rho(i,t-1)=\mu(1-\mu)^2\left(\hat\rho(i-1,t-1)+\hat\rho(i+1,t-1)-2\hat\rho(i,t-1)\right),\; |i|<N\\
&\hat\rho(-N,t)-\hat\rho(-N,t-1)=\mu(1-\mu)(\hat\rho(-N+1,t-1)-\hat\rho(-N,t-1))\\
&\hat\rho(N,t)-\hat\rho(N,t-1)=\mu(1-\mu)(\hat\rho(N-1,t-1)-\hat\rho(N,t-1))
\end{array}\right.
\label{diffusion}
\end{equation}
This is a discrete laplace equation in discrete time.  Given initial data,  the solution is unique for any positive integer $t$.
We state now our main result.
\begin{theorem} Let the $\{\sigma(k,i,0): (k,i)\in \cC_N\}$ be a set of independent Bernoulli random variables  and $\{\hat\rho_i:0<\hat \rho_i<1, \; i\in\Lambda_N\}$ such that
$\bbE(\sigma(k,i,0))=\hat\rho_i$\,, $\forall k \in \{1,\ldots,R\}$. Let  also $\hat\rho(\cdot,t)$ be the solution of the above system with initial condition $\hat\rho(i,0)=\hat\rho_i$, then
$\forall \epsilon>0$ and $\forall 0<\alpha<1$,
$$
\lim_{R\to\infty}\sup_{t\in[0,R^\alpha]}\bbP\left[\bigcup_{i=-N}^N\{|\rho^R(i,t)-\hat\rho(i,t)|>\epsilon\}\right]=0.
$$
\end{theorem}
\begin{proof}
By sub-additivity, it's enough to show :
\begin{equation}
\lim_{R\to\infty}\sup_{t\in[0,R^\alpha]}\bbP\left[\{|\rho^R(i,t)-\hat\rho(i,t)|>\epsilon\}\right]=0, \;i\in\Lambda_N.
\label{sub}
\end{equation}  
First, we show that
\begin{equation}
\bbE[\rho^R(i,t)]=\hat\rho(i,t),\;i\in\Lambda_N,\; 0<t<R^\alpha.
\label{average}
\end{equation}
We consider the cases $i\neq -N,N$, it is easy to see that the cases $i=-N,N$ are similar.
Using (\ref{recurevol}) and noting that :
$$
J(k-1,i)J(k-1,i-1)=0
$$
and 
$$
\bbE[J(k-1,i)]=\bbE[J(k-1,i-1)]=\mu(1-\mu)^2,\;\forall\; 1\leq k\leq R,
$$
we compute from (\ref{recurevol}) and (\ref{defJ}):
$$
\bbE[\rho^R(i,t)]-\bbE[\rho^R(i,t-1)]=\mu(1-\mu)^2\left(\bbE[\rho^R(i-1,t-1)]+\bbE[\rho^R(i+1,t-1)-2\bbE[\rho^R(i,t-1)]\right).
$$
We have used $t<R^\alpha<R$ so that $\sigma(k,i;t)$ and $\xi(k,j)$ are independent for all $k,i,j$ by Lemma (\ref{dependance}).  
Since $\bbE[\rho^R(i,0)]=\hat\rho_i$ and because the solution to (\ref{diffusion}) is unique, the relation (\ref{average}) is proven.
We now look at the variance of $\rho^R(i,t)$.
\begin{eqnarray}
{\rm Var}[\rho^R(i,t)]&=&\frac 1 {R^2}\bbE[\left(\sum_{k=1}^R\sigma(k,i;t)-\sum_{k=1}^R\bbE[\sigma(k,i;t)]\right)^2]\nonumber\\
&=&\frac 1{R^2}\left(\bbE[\sum_{k,k'=1}^R\sigma(k,i;t)\sigma(k',i;t)]-(\sum_{k=1}^R\bbE[\sigma(k,i;t)])^2\right)
\label{variance}
\end{eqnarray}
We can write for the second term :
\begin{eqnarray}
\bbE[\sigma(k,i;t)]&=&\sum_{x\in\cC_N}\bbE[\sigma(\tau^{-t}(k,i);0)|\tau^{-t}(k,i)=x]\bbP[\tau^{-t}(k,i)=x]\nonumber\\
&=&\sum_{x\in\cC_N}\bbE[\sigma(x;0)]\bbP[\tau^{-t}(k,i)=x].
\end{eqnarray}
Proceeding in the same way for the first one we obtain,
\begin{eqnarray}
\bbE[\sigma(k,i;t)\sigma(k',i;t)]&=&\sum_{x,x'\in\cC_N}\bbE[\sigma(x;0) \sigma(x';0)] \bbP[\tau^{-t}(k,i)=x,\tau^{-t}(k',i)=x'].\nonumber\\
\end{eqnarray}
When $k\neq k'$, we get :
\begin{eqnarray}
\bbE[\sigma(k,i;t)\sigma(k',i;t)]&=&\sum_{x\neq x'\in\cC_N}\bbE[\sigma(x;0)]\bbE[\sigma(x';0)] \bbP[\tau^{-t}(k,i)=x,\tau^{-t}(k',i]=x')\nonumber\\
\end{eqnarray}
because if $k\neq k'$, then $\tau^{-t}(k,i)\neq\tau^{-t}(k',i)$ for any configuration of scatterers, since $\tau$ (and $\tau^{-t}$) is a bijection.  The factorization of the expectation is obtained because the initial occupation variables are independent.
Going back to (\ref{variance}), we take the first expectation :
\begin{eqnarray}
\bbE[\sum_{k,k'=1}^R\sigma(k,i;t)\sigma(k',i;t)]&=&\sum_{k=1}^R\sum_{x\in\cC_N}\bbE[\sigma^2(x;0)]\bbP[\tau^{-t}(k,i)=x]\nonumber\\
&+& \sum_{k\neq k'\in\cC_N}\sum_{x\neq x'\in\cC_N}\bbE[\sigma(x;0)]\bbE[ \sigma(x';0)] \bbP[\tau^{-t}(k,i)=x,\tau^{-t}(k',i)=x']\nonumber\\
\label{firstterm}
\end{eqnarray}
Since $\bbE[\sigma^2(x;0)]\leq 1$, the first term above may be bounded by $R$ and we get :
\begin{equation}
\sum_{k,k'=1}^R\bbE[\sigma(k,i;t)\sigma(k',i;t)]\leq R+\sum_{k\neq k'}\sum_{x\neq x'\in\cC_N}\bbE[\sigma(x;0)]\bbE[ \sigma(x';0)] \bbP[\tau^{-t}(k,i)=x,\tau^{-t}(k',i)=x'].
\label{firsttermbound}
\end{equation}
For the second term of (\ref{variance}), we have :
\begin{eqnarray}
(\sum_{k=1}^R\bbE[\sigma(k,i;t)])^2
&\geq&\sum_{k\neq k'}\sum_{x,x'\in\cC_N}\bbE[\sigma(x;0)]\bbE[ \sigma(x';0)] \bbP[\tau^{-t}(k,i)=x]\bbP[\tau^{-t}(k',i)=x']\nonumber\\
\label{secondtermbound}
\end{eqnarray}
Combining (\ref{firsttermbound}) and (\ref{secondtermbound}), we get for the variance (\ref{variance})
\begin{eqnarray}
{\rm Var}[\rho^R(i,t)]\leq\frac 1 R+&&\frac 1{R^2}|\sum_{k\neq k'}\sum_{x,x'\in\cC_N}\bbE[\sigma(x;0)]\bbE[ \sigma(x';0)]\cdot\nonumber\\
&&\left( \bbP[\tau^{-t}(k,i)=x,\tau^{-t}(k',i)=x']-\bbP[\tau^{-t}(k,i)=x]\bbP[\tau^{-t}(k',i)=x']\right)|\nonumber\\
\end{eqnarray}
Using rotational invariance of both the distribution of the scatterers and the distribution of initial occupation variables, we get

\begin{eqnarray}
{\rm Var}[\rho^R(i,t)]\leq\frac 1 R+&&\frac 1{R}|\sum_{ k'\neq 1}\sum_{x,x'\in\cC_N}\bbE[\sigma(x;0)]\bbE[ \sigma(x';0)]\cdot\nonumber\\
&&\left( \bbP[\tau^{-t}(1,i)=x,\tau^{-t}(k',i)=x']-\bbP[\tau^{-t}(1,i)=x]\bbP[\tau^{-t}(k',i)=x']\right)|.\nonumber\\
\end{eqnarray}
Now, by Lemma (\ref{dependance}), if  $t+1<k'\leq R-t+1$ then $\tau^{-t}(0,i)$ and $\tau^{-t}(k',i)$ are independent random variables and for those $k'$, we have :
\begin{equation}
\bbP[\tau^{-t}(1,i)=x,\tau^{-t}(k',i)=x']-\bbP[\tau^{-t}(1,i)=x]\bbP[\tau^{-t}(k',i)=x']=0.
\label{recollisions}
\end{equation}

We are thus left with :
\begin{eqnarray}
{\rm Var}[\rho^R(i,t)]\leq&&\frac 1 R+\frac 1{R}\sum_{\begin{array}{ll} {\scriptstyle R-t+1< k'\leq R}\\{\scriptstyle  1<k'\leq  t+1}\end{array} }\sum_{x,x'\in\cC_N} \bbP[\tau^{-t}(1,i)=x,\tau^{-t}(k',i)=x']\nonumber\\
&+&\frac 1 R\sum_{\begin{array}{ll} {\scriptstyle R-t+1< k'\leq R}\\{\scriptstyle  1<k'\leq  t+1}\end{array} }\sum_{x,x'\in\cC_N}\bbP[\tau^{-t}(1,i)=x]\bbP[\tau^{-t}(k',i)=x']\nonumber\\
\leq && \frac 1 R+\frac {4(t-1)}{R}\nonumber\\
\leq&&\frac 6 {R^{1-\alpha}}, {\rm for } \;R\; {\rm large \; enough}.
\end{eqnarray}
In the first inequality, we have used independence and $\bbE[\sigma(x;0)]\leq 1$, in the second one the fact that probabilities sum up to $1$ and the fact that there are $2t-2$ terms in each sum over $k'$.  Finally, we have used the hypothesis that $t<R^\alpha$.  We conclude by using Chebyshev's inequality and get :
\begin{equation}
\bbP\left[\{|\rho^R(i,t)-\hat\rho(i,t)|>\epsilon\}\right]\leq \frac{6}{\epsilon^2 R^{1-\alpha}}.
\label{final}
\end{equation}
This in turn, implies (\ref{sub}).

\end{proof}
\section{Relation to Boltzmann equation}
By using a ``molecular chaos" hypothesis, one can derive heuristically a Boltzmann equation for the evolution  of the proportion of occupied sites on ring $\rho^R(i,t)$.  This equation coincides with the diffusion equation (\ref{diffusion}).  We start by an exact relation for the evolution of the number of particles contained in $\cR_i$, $|i|< N$.
\begin{eqnarray}
R\rho^R(i,t+1)-R\rho^R(i,t)&=&X_{(i+1)\to i}(t)-X_{i\to (i+1)}(t)\nonumber\\
&+&X_{(i-1)\to i}(t)-X_{i\to( i-1)}(t)\nonumber\\
\label{xdef}
\end{eqnarray}
$X_{i\to (i+1)}(t)$ counts the number of pairs of sites on rings $\cR_i$ and $\cR_{i+1}$ for which, at time $t$, a particle will jump from ring $\cR_i$ to $\cR_{i+1}$ and no particle jump from   $\cR_{i+1}$ to $\cR_{i}$.  
Namely,
\begin{equation}
X_{i\to i+1}(t)=|\{k\in\{1,\ldots,R\}: \sigma(k,i,t)=1,\;\sigma(k,i+1,t)=0, J(k,i)=1\}|.
\end{equation}
We define also
\begin{equation}
\hat X_{i\to i+1}(t)=|\{k\in\{1,\ldots,R\}: \sigma(k,i,t)=1,\;\sigma(k,i+1,t)=0\}|.
\end{equation}

All other $X$'s appearing in (\ref{xdef}) are defined in an analogous way. In this context, the molecular chaos hypothesis amounts to assume :
\begin{enumerate}
\item The proportion $\hat X_{i\to i+1}(t)/R$ is independent of the scatterer distribution between the rings and therefore :
\begin{eqnarray}
X_{i\to i+1}(t)/R&\simeq&\frac 1 R\sum_{k=1}^RJ(k,i)\; \hat X_{i\to i+1}(t)/R\nonumber\\
&\simeq&\bbE[J(k,i)]\; \hat X_{i\to i+1}(t)/R\nonumber\\
&\simeq&\mu(1-\mu)^2\hat X_{i\to i+1}(t)/R
\end{eqnarray}
as $R\to\infty$.
\item The proportion  $\hat X_{i\to i+1}(t)/R$ is given by the product of the proportions of occupied sites on $\cR_i$ and vacant sites on $\cR_{i+1}$ :
$$
\hat X_{i\to i+1}(t)/R=\rho^R(i,t)(1-\rho^R(i+1,t))
$$
\end{enumerate}
The first assumption is similar to the one made in the Kac model. The second one is more similar to a ``genuine" Boltzmann assumption on the factorization of distribution of pairs of particles.
One is thus led to write :
$$
X_{i\to (i+1)}(t)/R=\mu(1-\mu)^2\rho^R(i,t)(1-\rho^R(i+1,t)),\; i<N.
$$
A simple computation leads then to the evolution equation for $|i|<N$.
$$
\rho^R(i,t)-\rho^R(i,t-1)=\mu(1-\mu)^2\left(\rho^R(i-1,t-1)+\rho^R(i+1,t-1)-2\rho^R(i,t-1)\right).
$$
This is identical to (\ref{diffusion}) for $|i|<N$, equations for the evolution of the boundary densities are derived in a similar way.
\section{Conclusions}
We observe that taking a diffusive scaling limit  (when $N\to\infty$) of the discrete-time discrete-space diffusion equation (\ref{diffusion}) will yield the usual diffusion equation
\begin{equation}
\partial_t\rho(x,t)=\mu(1-\mu)^2\partial^2_x\rho(x,t),\; x\in[0,1], t\in\bbR_+,
\label{heat}
\end{equation}
with Neuman boundary conditions. A more challenging task would be to derive (\ref{heat}) for a fixed parameter $R$ directly in the diffusive scaling limit (when $N\to\infty$).  In that case, the absence of ``collisions" of sufficiently distant paths expressed in (\ref{recollisions}) is no more automatic and one should rely on a decay of spatio-temporal correlations.  In order to obtain such a decay, it might be necessary to define the model on a higher-dimensional lattice. More generally, the possibilities of  modifying the parameters of the model are numerous : one could also take correlated distribution of scatterers and study the effect of the correlations on the conduction properties of the model.  

We note  that it is easy to find interactions (namely a fixed distribution of scatterers) between rings such that the empirical densities do not follow the diffusion equation (\ref{diffusion}).  In contrast to the seemingly ``chaotic" orbits that are typical of systems obeying normal diffusion, those configurations of scatterers tend to have a long-range order and give rise to very ordered orbits.  In the set of possible interactions they represent ``integrable" dynamics.  By analogy, this might teach us a useful point to keep in mind  when studying systems described  by real Hamiltonian dynamics.  We tend to conceive and study models that are (random or not) perturbations of  dynamics that are ``simple" or lend themselves to computations .  But those systems might be rather untypical in the space of all possible interactions or indeed in real physical systems that obey the ordinary laws of macroscopic physics.

\vspace{5mm}
\noindent {\bf Acknowledgments.}  I thank Jean Bricmont and Carlos Mejia-Monasterio for discussions.
This work was supported by the ANR SHEPI grant.

\end{document}